\documentclass{article}
\usepackage{parskip,fullpage}
\usepackage{authblk}
\usepackage{amsmath,amsthm,amssymb,amsfonts}
\usepackage{physics}
\usepackage[colorlinks,citecolor=blue]{hyperref}
\usepackage{cleveref}

\usepackage[
	lambda,
	operators,
	advantage,
	sets,
	adversary,
	landau,
	probability,
	notions,	
	logic,
	ff,
	mm,
	primitives,
	events,
	complexity,
	asymptotics,
	keys]{cryptocode}

\renewcommand{\epsilon}{\varepsilon}
\newcommand{\F}{\mathbb{F}}

\newcommand{\from}{\gets}

\DeclareMathOperator{\calG}{\mathcal{G}}

\DeclareMathOperator{\calR}{\mathcal{R}}

\DeclareMathOperator{\calX}{\mathcal{X}}
\DeclareMathOperator{\calY}{\mathcal{Y}}

\DeclareMathOperator{\regQ}{\mathsf{Q}}
\DeclareMathOperator{\regR}{\mathsf{R}}

\newcommand{\Gbbotp}{\calG_{\mathrm{BB}\text{-}\mathrm{OTP}}}
\newcommand{\Gcollapsing}{\calG_{\mathrm{Collapsing}}}
\newcommand{\Gauth}{\calG_{\mathrm{Auth}}}

\newcommand{\Obf}{\mathsf{Obf}}

\newcommand{\KeyGen}{\mathsf{KeyGen}}
\newcommand{\TokenGen}{\mathsf{TokenGen}}
\newcommand{\TokenEval}{\mathsf{TokenEval}}
\newcommand{\AuthKeyGen}{\mathsf{AuthKeyGen}}
\newcommand{\AuthTokenGen}{\mathsf{AuthTokenGen}}
\newcommand{\Sign}{\mathsf{Sign}}
\newcommand{\Verify}{\mathsf{Verify}}

\newcommand{\tk}{\mathsf{tk}}
\newcommand{\Increase}{\mathsf{Increase}}
\newcommand{\StdDecomp}{\mathsf{StdDecomp}}
\newcommand{\FindInput}{\mathsf{FindInput}}
\newcommand{\Invert}{\mathsf{Invert}}
\newcommand{\CPhsO}{\mathsf{CPhsO}}

\newtheorem{theorem}{Theorem}
\newtheorem{proposition}{Proposition}

\newtheorem{lemma}{Lemma}

\newtheorem{claim}{Claim}

\newtheorem*{remark*}{Remark}

\makeatletter
\newtheorem*{rep@lemma}{\rep@title}
\newcommand{\newreptheorem}[2]{%
\newenvironment{rep#1}[1]{%
 \def\rep@title{#2 \ref{##1}}%
 \begin{rep@lemma}}%
 {\end{rep@lemma}}}
\makeatother

\newreptheorem{lemma}{Lemma}

\theoremstyle{definition}
\newtheorem{definition}{Definition}
\newtheorem{construction}{Construction}

\Crefname{fact}{Fact}{Facts}

\title{Quantum One-Time Protection of any Randomized Algorithm}

\author{Sam Gunn and Ramis Movassagh}
\affil{Google Quantum AI, Venice CA, 90291}

\date{\today}

\begin{document}

\maketitle

\begin{abstract}
    The meteoric rise in power and popularity of machine learning models dependent on valuable training data has reignited a basic tension between the power of running a program locally and the risk of exposing details of that program to the user.  
    At the same time, fundamental properties of quantum states offer new solutions to data and program security that can require strikingly few quantum resources to exploit, and offer advantages outside of mere computational run time.
    In this work, we demonstrate such a solution with \emph{quantum one-time tokens}.
    
    A quantum one-time token is a quantum state that permits a certain program to be evaluated \emph{exactly once}.
    One-time security guarantees, roughly, that the token cannot be used to evaluate the program more than once.
    We propose a scheme for building quantum one-time tokens for any randomized classical program, which include generative AI models.
    We prove that the scheme satisfies an interesting definition of one-time security \emph{as long as outputs of the classical algorithm have high enough min-entropy}, in a black box model.
    
    Importantly, the classical program being protected does not need to be implemented coherently on a quantum computer. 
    In fact, the size and complexity of the quantum one-time token is \emph{independent} of the program being protected, and additional quantum resources serve only to increase the security of the protocol.  
    Due to this flexibility in adjusting the security, we believe that our proposal is parsimonious enough to serve as a promising candidate for a near-term useful demonstration of quantum computing in either the NISQ or early fault tolerant regime.
\end{abstract}

\section{Introduction}
Commercializing software presents a central dilemma: How does a proprietor distribute software, without forfeiting ownership?
On the one hand, software must be made available to users for them to use it; on the other, once the software is distributed, an unlimited number of unauthorized copies can be made.
This problem is more acute in the age of generative AI, where the software can be extremely valuable and potentially reveal private information.

There are two widely-adopted solutions to this problem:
\begin{enumerate}
    \item The proprietor can distribute an obfuscated version of the software. However, this opens up the possibility of piracy, and the user can always run the software an unlimited number of times.
    \item The proprietor can allow queries to the software, instead of distributing the software itself. But this requires communication between the user and the proprietor every time the software is used.
\end{enumerate}
Sometimes a combination of these solutions is employed --- for instance, an internet browser might be obfuscated and made public while the search engine itself is kept on servers.

There is an apparent trade-off between usability and exclusivity in these solutions.
Distributing obfuscated software makes it highly usable, but it is impossible to impose any limits on the number of times it can be used.
Keeping the software on servers and responding to user queries has high exclusivity, but lower usability because the user needs to communicate with the server to perform computations.
Furthermore, this solution is only available to proprietors with enough resources to host a reliable server.

This trade-off between usability and exclusivity is inherent in the classical setting, because classical information can always be copied.
As soon as software is distributed, the user can query or copy it as many times as they want.
Therefore it is natural to look to quantum mechanics for improved solutions. In contrast to classical information, quantum information cannot be cloned.
Indeed, unclonable cryptography is able to make use of this principle to improve both of the above solutions.\footnote{We note that unclonable cryptography has led to many additional interesting and varied protocols, including quantum money \cite{Wie83}, quantum key distribution \cite{Ben14}, and certified deletion \cite{BI20}.}

\paragraph{Improving Solution (1) with copy protection.}
In copy protection, introduced by \cite{Aar09}, a quantum state is created that allows a program to be run an unlimited number of times, but not copied.
In a classical oracle model, quantum copy protection is possible for any unlearnable program \cite{ALL+21}.
In the standard model, there exist unlearnable programs that cannot be copy protected \cite{ALP21}; nonetheless, it is known how to copy protect a small number of particular functionalities in the standard model \cite{CLLZ21,LLQZ22,CG24,AB24}.

Quantum copy protection addresses the piracy concern in Solution (1), but it does not apply in the setting where the proprietor wishes to distribute \emph{individual queries} to the software.
This is particularly relevant in the current age of generative AI, where the pay-per-query model is prevalent.

\paragraph{Improving Solution (2) with one-time programs.}
For the pay-per-query model, one-time programs are a more suitable solution.
A one-time program token is a quantum state that enables a program to be evaluated \emph{once}.
One-time programs therefore remove the need for online communication in Solution (2): The proprietor distributes tokens, and the users consume the tokens offline at their leisure.

\paragraph{This work.}
In this work, we propose the first general-purpose tokenized program scheme using quantum information.
Prior work was restricted to protecting specific cryptographic functionalities, and in particular did not apply to non-cryptographic programs like generative AI models.

Our method uses a one-time signature scheme and a program obfuscator to compile any randomized algorithm into a one-time token that allows the algorithm to be executed on any input of the user's choice.
We note that the quantum capabilities required for our scheme are \emph{completely independent of the program being protected}.
Therefore, our scheme could plausibly be used to protect a large program using only a small, noisy quantum device --- even one that is incapable of demonstrating quantum computational advantage!

\paragraph{Prior and concurrent work on one-time programs.}
One-time programs were initially studied in \cite{GKR08}, but without quantum computation.
The idea of quantum one-time programs was originally suggested in \cite{BGS13}, where it was shown that it is impossible to build quantum one-time tokens for deterministic programs without specialized hardware assumptions.
This is because a simple rewinding attack would allow a user to make arbitrarily many queries to the algorithm, given any state that allows it to be run once.
Later, \cite{BS23} presented a scheme to one-time protect signature functions in the standard model.

In concurrent and independent work, \cite{GLR+24} also present a scheme for the one-time protection of arbitrary randomized algorithms.
Their scheme is essentially identical to ours, except that they instantiate the quantum one-time authentication scheme with the particular construction of \cite{CLLZ21}.
With respect to security definitions and proofs, their results are generally much more extensive, including stronger (albeit more complex) definitions of one-time security and new impossibility results.
However, they take a very different and complementary approach to proving security.
It would be interesting to know whether our results (and in particular \Cref{lemma:collapsing}, which is simple and self-contained) say anything about security in their setting, or whether their results imply our \Cref{lemma:collapsing}.

\paragraph{Removing quantum communication.}
The protocol described below requires quantum communication between the proprietor and the user, because the proprietor needs to send the user the quantum one-time program token.
However, if the one-time signatures are instantiated with those from \cite{CLLZ21}, then the results of \cite{Shm22,CHV23} allow us to replace this quantum communication with a remote state preparation protocol that uses only classical communication.
That is, there exists a polynomial-time interactive protocol between a quantum user and a \emph{classical} proprietor that results in the user holding a quantum one-time token which they can only use once.
Unfortunately, this protocol is highly complex and not likely to be feasible on a near-term quantum device.\footnote{It also requires the assumptions that LWE is sub-exponentially hard for quantum computers and that indistinguishability obfuscation for classical circuits exists with sub-exponential security against quantum polynomial-time adversaries.}


\section{The construction}
In this section, we present our scheme for one-time protecting an arbitrary randomized algorithm.
Any randomized algorithm can be described by a function $f$, which takes the input together with a random string, and outputs the result of the computation.

Our scheme relies on three cryptographic primitives:
\begin{itemize}
    \item $(\AuthKeyGen, \AuthTokenGen, \Sign, \Verify)$, a quantum one-time authentication scheme. For instance, any quantum one-time signature scheme will suffice --- although it is not required to be publicly verifiable. We define a quantum one-time authentication scheme in \Cref{definition:one-time-auth}.
    \item $\Obf$, a classical program obfuscator. This can be heuristically instantiated with any classical obfuscation algorithm (for instance, a candidate indistinguishability obfuscator), but we model it as black-box obfuscation for our security proof.
    \item $H$, a hash function modeled as a random oracle. This can be replaced with a pseudorandom function without affecting the security proof, because we model $\Obf$ as black-box obfuscation.
\end{itemize}

The software proprietor will generate a program token which can be used to evaluate $f$ one time. A program token consists of two parts:
\begin{itemize}
    \item $\ket{\psi}$, an unclonable quantum state that does not depend on $f$.
    \item $\Obf(P)$, an obfuscated classical circuit $P$ that depends on $f$. This part is completely classical.
\end{itemize}

We emphasize that $\ket{\psi}$ \emph{does not depend on the program being obfuscated}.
In particular, this means that the quantum capabilities required of the user and the software proprietor do not depend on the complexity of the computation being one-time protected.
Highly-complex computations like evaluation of a generative AI model can be protected, even if the quantum capabilities of both the user and software proprietor are limited.
Furthermore, if we use the one-time signature scheme from \cite{CLLZ21} for one-time authentication, then the only quantum capabilities required of the user are to store a constant-sized quantum state and to measure it in the standard and Hadamard bases.

At a high level, our scheme works by requiring the user to one-time authenticate any input they wish to evaluate $f$ on.
The program $P$ will not evaluate $f$ unless a valid input-authentication pair is provided.
Crucially, $P$ will determine the randomness to use for $f$ by computing a hash of the input-authentication pair.

\subsection{One-time authentication with strong unforgeability}
Before we state our one-time program construction in \Cref{construction:one-time-programs}, let us define quantum one-time authentication and strong unforgeability.
This is essentially a secret-key version of the definition of a quantum one-time signature scheme from \cite{BS23}, except that we insist on it being infeasible to generate two distinct signatures \emph{even for the same message}.
This property is referred to as \emph{strong} unforgeability in \cite{CHV23,BKNY23}.

\begin{remark*}
    There was an error in a previous version of this paper, due to our defining only the standard, weak version of unforgeability. We thank an anonymous reviewer for QIP 2025 for pointing this out.
\end{remark*}

In the following definition, the notation $\adv^{\Verify(\sk, \cdot)}$ means that $\adv$ is given quantum query access to the function $(x,z) \mapsto \Verify(\sk, (x,z))$.

\begin{definition}[Quantum one-time authentication with strong unforgeability] \label{definition:one-time-auth}
    We say that a tuple of four polynomial-time algorithms $(\AuthKeyGen, \AuthTokenGen, \Sign, \Verify)$ is a \emph{quantum one-time authentication scheme with strong unforgeability} if the following hold:\footnote{The notation $\secparam$ just means $\secpar$ repeated 1's; we sometimes use $\secparam$ instead of $\secpar$ as the input to an algorithm because we want the algorithm to run in time that is polynomial in the size of the input.}
    \begin{itemize}
        \item (Correctness) $\Verify(\sk, (x,z)) = 1$ if $z \from \Sign(x, \ket{\tk})$, $\ket{\tk} \from \TokenGen(\sk)$, and $\sk \from \AuthKeyGen(\secparam)$.
        \item (One-time unforgeability) For any polynomial-time adversary $\adv$,
        \[
            \Pr_{\substack{(x_1, z_1, x_2, z_2) \from \adv^{\Verify(\sk, \cdot)}(\ket{\tk}) \\ \ket{\tk} \from \AuthTokenGen(\sk) \\ \sk \from \AuthKeyGen(\secparam)}}[\Verify(\sk, (x_1,z_1)) = \Verify(\sk, (x_2,z_2)) = 1 \text{ and } (x_1,z_1) \ne (x_2,z_2)] = \negl.
        \]
    \end{itemize}
\end{definition}

Since (strongly unforgeable) one-time authentication schemes can be easily built from any (strongly unforgeable) one-time signature scheme by simply considering the public key to be part of the secret key, the results of \cite{BS23,CLLZ21,CHV23,BKNY23} imply their existence.
In fact, these results show the existence of \emph{information-theoretically secure} quantum one-time authentication schemes --- that is, these constructions are secure against computationally unbounded adversaries, as long as they are only allowed to make $\poly$ queries to $\Verify$.

We describe one such scheme for single-bit messages now.
The case for multi-bit messages works by simply using several single-bit schemes in parallel.

\begin{construction}[Single-bit quantum one-time authentication from hidden subspace states, \cite{BS23}] \label{construction:one-time-auth}
    The algorithms are defined as follows.
    
    \paragraph{$\AuthKeyGen(\secparam)$}
    \begin{enumerate}
        \item Sample a uniformly random subspace $A \subseteq \F_2^\secpar$ of dimension $\secpar/2$.
        \item Output $A$.
    \end{enumerate}

    \paragraph{$\AuthTokenGen(A)$}
    \begin{enumerate}
        \item Output $\ket{A} = 2^{-\secpar/4} \sum_{a \in A} \ket{a}$.
    \end{enumerate}

    \paragraph{$\Sign(x, \ket{A})$}
    \begin{enumerate}
        \item If $x = 0$, measure $A$ in the standard basis and output the result.
        \item If $x = 1$, measure $A$ in the Hadamard basis and output the result.
    \end{enumerate}

    \paragraph{$\Verify(A, (x,z))$}
    \begin{enumerate}
        \item If $z = 0^\secpar$, reject (output $\bot$).
        \item If $x = 0$ and $z \in A$, accept (output 1).
        \item If $x = 1$ and $z \in A^\perp$, accept (output 1).
        \item Otherwise, reject (output $\bot$).
    \end{enumerate}
\end{construction}

Note that \Cref{construction:one-time-auth} is correct because
\[
    H^\secpar \ket{A} = 2^{-\secpar/4} \sum_{\substack{b \in \F_2^n \\ b \cdot a = 0 \ \forall a \in A}} \ket{b} = \ket*{A^\perp}.
\]

Security was proven in \cite{BS23}.

\begin{theorem}[Adapted from Theorem 16 of \cite{BS23}]
    \Cref{construction:one-time-auth} is an information-theoretically secure quantum one-time authentication scheme.
\end{theorem}

\subsection{Our one-time program construction}

Let $f$ be the program we wish to one-time protect.
Our scheme works by publishing a program $\hat{P}$ which contains within it the description of $f$.
Crucially, this program will have the following properties:
\begin{itemize}
    \item $\hat{P}$ uses obfuscation so as not to reveal the code of $f$;
    \item $\hat{P}$ will compute $f$ on any input, but only if the input is authenticated with a one-time authentication scheme; and
    \item $\hat{P}$ suitably ``mixes'' the input, the authentication tag, and the output together in such a way that, if the output is measured, the input and authentication tag effectively collapse.
\end{itemize}

While the first property follows from our use of an oracle model, and the second property is by construction, the third will prove to be quite difficult to establish.
The statement of this third property is formalized in \Cref{lemma:collapsing}, which allows us to prove the one-time security of our scheme.

\begin{construction}[General-purpose one-time programs] \label{construction:one-time-programs}
    Let $f : \calX \times \calR \to \calY$ be the function to be one-time protected, described as a bit-string representing a classical circuit that computes $f$.
    Let $\Obf$ be an obfuscation algorithm for classical circuits, and let $(\AuthKeyGen, \AuthTokenGen, \Sign, \Verify)$ be a quantum one-time authentication scheme as defined in \Cref{definition:one-time-auth}.
    Suppose that the one-time authentication scheme produces authentications of length $m$ for all messages $x \in \calX$, and let $H : \calX \times \{0,1\}^m \to \calR$ be a random oracle.
        
    Our one-time protection scheme is specified by the algorithms $\KeyGen$, $\TokenGen$, and $\TokenEval$, defined as follows.

    \paragraph{$\KeyGen(\secparam, f)$}
    \begin{enumerate}
        \item Sample $\sk \from \AuthKeyGen(\secparam)$.
        \item Let $P : \calX \times \{0,1\}^m \to \calY \cup \{\bot\}$ be a classical circuit that does the following on input $(x, z)$:
        \begin{enumerate}
            \item Compute $\Verify(\sk, (x, z))$. If it rejects, output the special ``reject'' symbol $\bot$.
            \item Otherwise, output $f(x; H(x, z))$.
        \end{enumerate}
        \item Compute the obfuscation $\hat{P} = \Obf(P)$ of $P$.
        \item Output $(\sk, \hat{P})$.
    \end{enumerate}
    
    \paragraph{$\TokenGen(\sk)$}
    \begin{enumerate}
        \item Compute the one-time authentication token $\ket{\tk} \from \AuthTokenGen(\sk)$.
        \item Output $\ket{\tk}$ as the one-time program token.
    \end{enumerate}

    \paragraph{$\TokenEval(x, \ket*{\tk}, \hat{P})$}
    \begin{enumerate}
        \item Compute $z \from \Sign(x, \ket*{\tk})$.
        \item Compute $\hat{P}(x, z)$ and output the result.
    \end{enumerate}
\end{construction}

After $\KeyGen$ is run, the obfuscated program $\hat{P} = \Obf(P)$ is published as a public key.
Given a one-time program token $\ket{\tk}$, together with this public key, a user can evaluate $f$ on any input $x$ of their choice.

\section{Security} \label{section:security}
We formalize one-time security of our one-time program scheme with the following \emph{black-box one-time program} game.
Essentially, security says that once an adversary produces a measured output of $\hat{P}$, the adversary cannot find new accepting inputs to $\hat{P}$.

\begin{enumerate}
    \item[] $\Gbbotp(\ket{\psi}, \hat{P})$:
    \item The adversary is given $\ket{\psi}$. The adversary is also given (quantum) oracle access to $\hat{P}$. The adversary is allowed to access this oracle throughout the game.
    \item The adversary submits a quantum query to the challenger on register $\regQ$. The challenger does the following:
    \begin{enumerate}
        \item Compute $\hat{P}$ on $\regQ$, placing the result of the computation onto a register $\regR$.
        \item Measure $\regR$, obtaining outcome $y$.
        \item If $y = \bot$, the game is aborted and the adversary loses.
        \item Otherwise, return $\regQ$ to the adversary.
    \end{enumerate}
    \item The adversary submits a quantum query to the challenger on register $\regQ$. The challenger does the following:
    \begin{enumerate}
        \item Compute $\hat{P}$ on $\regQ$, placing the result of the computation onto a register $\regR$.
        \item Measure $\regR$, obtaining outcome $y'$.
    \end{enumerate}
    \item The adversary wins if $y' \not\in \{y, \bot\}$.
\end{enumerate}

We say that a scheme is black-box one-time secure if no polynomial-time adversary can win $\Gbbotp$ with inverse polynomial probability.

\begin{definition}[One-time program] \label{definition:one-time-security}
    A one-time program scheme consists of three polynomial-time algorithms $\KeyGen$, $\TokenGen$, and $\TokenEval$.
    We say that it is \emph{black-box one-time secure} for a function $f$ if, for all polynomial-time adversaries $\adv$ making at most $\poly$ quantum oracle queries,
    \[
        \Pr_{\substack{(\sk, \hat{P}) \from \KeyGen(\secparam, f) \\ \ket{\tk} \from \TokenGen(\sk)}}[\adv^{\hat{P}} \text{ wins } \Gbbotp(\ket{\tk}, \hat{P})] = \negl.
    \]
\end{definition}

We do not believe that this is the final say in definitions of one-time security, but we believe it is a useful starting point.
The rest of this paper is devoted to proving the following theorem.

\begin{theorem}[\Cref{construction:one-time-programs} is one-time secure] \label{theorem:one-time-security}
    Suppose that, for every $x \in \calX$, the min-entropy of $f(x; r)$ for random $r \from \calR$ is at least $\poly$.
    Suppose further that the quantum one-time authentication scheme in \Cref{construction:one-time-programs} is secure.
    Then \Cref{construction:one-time-programs} is black-box one-time secure for $f$.
\end{theorem}

Our proof of \Cref{theorem:one-time-security} will center on the notion of a \emph{collapsing hash function}.
The definition of a collapsing hash function is due to \cite{Unr16}; informally, we say that $g$ is collapsing if measuring the output of $g$ collapses the input register from the perspective of any polynomial-time adversary.

Formally, let $g^H$ be a function that is computed by making calls to a random oracle $H$.
In this case we define collapsing with the following game.

\begin{enumerate}
    \item[] $\Gcollapsing(b, g^H)$:
    \item The adversary is given a full description of $g$ and oracle access to $H$.
    \item The adversary sends the challenger a query on register $\regQ$.
    \item The challenger computes $g^H$ on register $\regQ$ and measures the outcome. If $b = 1$, the challenger measures $\regQ$ as well.
    \item The challenger returns $\regQ$ to the adversary.
    \item The adversary outputs a bit $b'$.
\end{enumerate}

\begin{definition}[Collapsing hash function \cite{Unr16}]
    Let $g^H$ be a function that is computed by making calls to a random oracle $H$.
    We say that $g^H$ is \emph{collapsing} if, for all adversaries $\adv$ making at most $\poly$ quantum oracle queries to $H$,
    \[
        \Pr_{\substack{H \\ b \from \{0,1\}}}[\adv^H \text{ outputs } b' = b \text{ in } \Gcollapsing(b,g^H)] \le \frac{1}{2} + \negl.
    \]
\end{definition}

The key idea in our proof of \Cref{theorem:one-time-security} is to show that the function $g : x \mapsto f(x; H(x))$ is collapsing if $H$ is a random function and $f$ has significant min-entropy for every $x$.
This is stated as \Cref{lemma:collapsing}.

\begin{lemma} \label{lemma:collapsing}
    Let $f : \{0,1\}^m \times \{0,1\}^n \to \calY$. Suppose that for all $x \in \{0,1\}^m$, the min-entropy of $f(x; r)$ for random $r \from \{0,1\}^n$ is at least $\tau$. Then if $H : \{0,1\}^m \to \{0,1\}^n$ is a random oracle, the function $g^H : x \mapsto f(x; H(x))$ is collapsing with advantage $O(q^3 \cdot 2^{-\tau})$.
\end{lemma}

We will use the compressed oracles technique of \cite{Zha19} to prove \Cref{lemma:collapsing} in \Cref{section:proof-of-lemma-1}.
This proof is somewhat involved, so before we present it we will show in \Cref{section:warm-up} that random oracles are collapsing as a warm-up.
The proof of \Cref{lemma:collapsing} is essentially identical, except that the components need to be updated.

Once we have \Cref{lemma:collapsing}, \Cref{theorem:one-time-security} will follow immediately, as we will now see.

\begin{proof}[Proof of \Cref{theorem:one-time-security}]
    Suppose that an adversary $\adv$ making $\poly$ many oracle queries satisfies
    \[
        \Pr_{\substack{(\sk, \hat{P}) \from \KeyGen(\secparam, f) \\ \ket{\tk} \from \TokenGen(\sk)}}[\adv^{\hat{P}} \text{ wins } \Gbbotp(\ket{\tk}, \hat{P})] = \varepsilon
    \]
    for some $\varepsilon > 0$.
    We will use \Cref{lemma:collapsing} to give a reduction $\calR$ making $\poly$ many queries to $\Verify$ such that
    \[
        \Pr_{\substack{\sk \from \AuthKeyGen(\secparam) \\ \ket{\tk} \from \AuthTokenGen(\sk)}}[\calR^{\Verify} \text{ wins } \Gauth(\ket{\tk})] \ge \varepsilon - \negl,
    \]
    which will complete the proof.

    The reduction behaves as follows.
    \begin{enumerate}
        \item[] $\calR$:
        \item Receive $\ket{\tk}$ from the challenger and forward it to $\adv$. Every time the adversary makes a query to $P$, use the oracle access to $\Verify(\sk, \cdot)$ to compute the response.
        \item Receive $\regQ$ from the adversary. Measure $\regQ$, obtaining outcome $(x_1,z_1)$; return the collapsed register $\regQ$ to the adversary.
        \item Receive $\regQ$ from the adversary again. Measure $\regQ$, obtaining outcome $(x_2,z_2)$.
        \item Output $(x_1,z_1,x_2,z_2)$.
    \end{enumerate}
    
    The only thing that's different about the adversary's view in $\Gbbotp$ and in the game with $\calR$ is that $\regQ$ is measured in the game with $\calR$.
    But by \Cref{lemma:collapsing}, this is indistinguishable to the adversary!

    Now observe that if $\adv$ wins $\Gbbotp$, then $\calR$ breaks the one-time unforgeability of the one-time authentication scheme $(\AuthKeyGen, \AuthTokenGen, \Sign, \Verify)$.
    This is because $y \ne y'$ both must not be $\bot$, which means that $\Verify(\sk, (x_1,z_1)) = \Verify(\sk, (x_2,z_2)) = 1$ and $(x_1,z_1) \ne (x_2,z_2)$.
    This completes the proof.
\end{proof}

It only remains to prove \Cref{lemma:collapsing}.

\subsection{Warm-up: Random oracles are collapsing} \label{section:warm-up}
As a warm-up to proving \Cref{lemma:collapsing}, we use the compressed oracles technique to present a direct proof that random oracles are collapsing.
It is already known that random oracles are collapsing by a proof of \cite{Unr16}, but that proof makes use of multiple lemmas that are proven using the compressed oracle technique --- here, we prove the result directly using this technique.

\begin{proposition} \label{prop:random-oracles-collapse}
    A random oracle $R : \{0,1\}^m \to \{0,1\}^n$ is collapsing with advantage $O(q^3 / 2^n)$, where $q$ is the number of queries made by the adversary.
\end{proposition}
\begin{proof}
    We use the compressed oracle technique due to \cite{Zha19}, and we define $\CPhsO, \CPhsO'$ and the notation for compressed oracle databases as in that paper.
    Suppose that the adversary makes $q_0$ queries before the challenge query and $q_1$ queries after.
    Depending on whether or not the challenger records the challenge query input, the state just after the challenge query is
    \[
        \sum_{D, x} \alpha_{D, x} \ket{\psi_{D, x}} \otimes \ket{x} \otimes \ket{D, D(x)}
    \]
    or
    \[
        \sum_{D, x} \alpha_{D, x} \ket{\psi_{D, x}} \otimes \ket{x} \otimes \ket{D, D(x), x}.
    \]
    Here, the first register contains the adversary's state; the second is the query register; the third is the database register; and the fourth and fifth are the registers into which the challenger records the challenge query output and (in the second case) input.

    Next, the adversary makes $q_1$ further queries to $\CPhsO$.
    Let $U$ be the unitary describing the evolution of the game state under these queries.

    Let $\Invert$ be the unitary that reads $\ket{D, y}$ and writes the lexicographically first $x$ such that $(x, y) \in D$ onto a new register, if there is any such $x$.
    This function is almost identical to $\FindInput$ from \cite{Zha19}; one can compute either of $\Invert$ or $\FindInput$ with one call to the other.

    Our proof strategy is to show that applying $\Invert$ to $\ket{D, D(x)}$ at the end of the game in the $b=0$ case yields the state in the $b=1$ case.
    Since $\Invert$ is applied only to the challenger's side, this will complete the proof.

    Since $\Invert$ commutes with $\CPhsO'$, \cite[Lemma 7]{Zha19} implies that $\Invert$ $O(1/2^n)$-almost commutes with $\CPhsO$. Therefore
    \[
        \norm{\Invert \circ U \sum_{D, x} \alpha_{D, x} \ket{\psi_{D, x}} \otimes \ket{x, D, D(x)} - U \circ \Invert \sum_{D, x} \alpha_{D, x} \ket{\psi_{D, x}} \otimes \ket{x, D, D(x)}} = O\left(q / \sqrt{2^n}\right).
    \]
    By \cite[Theorem 2]{Zha19}, the mass on branches $D$ with collisions is at most $O(\sqrt{q^3/2^n})$ in the state on the right, so:
    \[
        \norm{U \circ \Invert \sum_{D, x} \alpha_{D, x} \ket{\psi_{D, x}} \otimes \ket{x, D, D(x)} - U \sum_{D, x} \alpha_{D, x} \ket{\psi_{D, x}} \otimes \ket{x, D, D(x), x}} = O\left(\sqrt{q^3 / 2^n}\right).
    \]
    Putting it together, we have
    \[
        \norm{\Invert \circ U \sum_{D, x} \alpha_{D, x} \ket{\psi_{D, x}} \otimes \ket{x, D, D(x)} - U \sum_{D, x} \alpha_{D, x} \ket{\psi_{D, x}} \otimes \ket{x, D, D(x), x}} = O\left(\sqrt{q^3 / 2^n}\right),
    \]
    completing the proof.
\end{proof}

\subsection{Proof of Lemma \ref{lemma:collapsing}} \label{section:proof-of-lemma-1}

In this section we will prove \Cref{lemma:collapsing}, which we reproduce below.

\begin{replemma}{lemma:collapsing}
    Let $f : \{0,1\}^m \times \{0,1\}^n \to \calY$. Suppose that for all $x \in \{0,1\}^m$, the min-entropy of $f(x; r)$ for random $r \from \{0,1\}^n$ is at least $\tau$. Then if $H : \{0,1\}^m \to \{0,1\}^n$ is a random oracle, the function $g : x \mapsto f(x; H(x))$ is collapsing with advantage $O(q^3 \cdot 2^{-\tau})$.
\end{replemma}
\begin{proof}
    We will follow the same proof strategy as \Cref{prop:random-oracles-collapse}, except that we will need to generalize the components to handle the case where the high-entropy function $f$ is applied to the random oracle output.
    We give our generalizations of \cite[Lemma 7]{Zha19} and \cite[Theorem 2]{Zha19} in \Cref{claim:almost-commuting} and \Cref{claim:collision-finding}, respectively.
    
    Suppose that the adversary makes $q_0$ queries before the challenge query and $q_1$ queries after.
    Depending on whether or not the challenger records the challenge query input, the state just after the challenge query is
    \[
        \sum_{D, x} \alpha_{D, x} \ket{\psi_{D, x}} \otimes \ket{x} \otimes \ket{D, f(x; D(x))}
    \]
    or
    \[
        \sum_{D, x} \alpha_{D, x} \ket{\psi_{D, x}} \otimes \ket{x} \otimes \ket{D, f(x; D(x)), x}.
    \]
    The first register contains the adversary's state; the second is the query register; the third is the database register; and the fourth and fifth are the registers into which the challenger records the challenge query output and (in the second case) input.

    Next, the adversary makes $q_1$ further queries to $\CPhsO$.
    Let $U$ be the unitary describing the evolution of the game state under these queries.

    Let $\Invert_f$ be the unitary that, given $D, y$, writes the lexicographically first $x$ such that $f(x; D(x)) = y$ onto a new register, if such an $x$ exists.
    Our proof strategy is to show that applying $\Invert_f$ to $\ket{D, f(x; D(x))}$ at the end of the game in the $b=0$ case yields the state in the $b=1$ case.
    Since $\Invert_f$ is applied only to the challenger's side, this will complete the proof.

    By \Cref{claim:almost-commuting},
    \begin{align*}
        \Biggl|\!\Biggl| \Invert_f \circ U \sum_{D, x} \alpha_{D, x} \ket{\psi_{D, x}} \otimes \ket{x, D, f(x; D(x))} & \\
        - U \circ \Invert_f \sum_{D, x} \alpha_{D, x} \ket{\psi_{D, x}} \otimes \ket{x, D, f(x; D(x))} \Biggr|\!\Biggr| &= O\left(q / \sqrt{2^\tau}\right).
    \end{align*}
    By \Cref{claim:collision-finding}, the mass on branches $D$ with collisions $f(x; D(x)) = f(x'; D(x'))$ is at most $O(\sqrt{q^3/2^\tau})$ in the state on the right, so
    \begin{align*}
        \Biggl|\!\Biggl| U \circ \Invert_f \sum_{D, x} \alpha_{D, x} \ket{\psi_{D, x}} \otimes \ket{x, D, f(x; D(x))} & \\
        - U \sum_{D, x} \alpha_{D, x} \ket{\psi_{D, x}} \otimes \ket{x, D, f(x; D(x)), x} \Biggr|\!\Biggr| &= O\left(\sqrt{q^3 / 2^\tau}\right).
    \end{align*}
    Putting it together, we have
    \begin{align*}
        \Biggl|\!\Biggl| \Invert_f \circ U \sum_{D, x} \alpha_{D, x} \ket{\psi_{D, x}} \otimes \ket{x, D, f(x; D(x))} & \\
        - U \sum_{D, x} \alpha_{D, x} \ket{\psi_{D, x}} \otimes \ket{x, D, f(x; D(x)), x} \Biggr|\!\Biggr| &= O\left(\sqrt{q^3 / 2^\tau}\right),
    \end{align*}
    completing the proof.
\end{proof}

We need a generalization of Theorem 2 from \cite{Zha19}.

\begin{claim} \label{claim:collision-finding}
    For any adversary making $q$ queries to $\CPhsO$, if the database $D$ is measured after the $q$ queries, the resulting database will contain distinct inputs $x, x'$ such that $f(x; D(x)) = f(x'; D(x'))$ with probability at most $O(q^3/2^\tau)$.
\end{claim}
\begin{proof}
    We closely follow the proof of Theorem 2 from \cite{Zha19}.
    
    We say that $D$ contains an $f$-\emph{collision} if it contains distinct inputs $x, x'$ such that $f(x; D(x)) = f(x'; D(x'))$.
    The compressed oracle's database starts out empty, so the probability of containing an $f$-collision starts out as 0.
    We will show that the probability cannot rise too much with each query.
    Consider the game state just before the $q$th query:
    \[
        \ket{\psi} = \sum_{x,w,z,D} \alpha_{x,w,D} \ket{x,w,z} \otimes \ket{D},
    \]
    where $D$ is the compressed phase oracle, $x, w$ are the query registers, and $z$ is the adversary's internal register.

    Let $P$ be the projection onto the span of basis states $\ket{x,w,z} \otimes \ket{D}$ where $D$ contains an $f$-collision. We will relate $\norm{P \ket{\psi}}$ to $\norm{P \circ \CPhsO \ket{\psi}}$.

    Let $Q$ be the projection onto the span of basis states $\ket{x,w,z} \otimes \ket{D}$ such that $D$ does not contain an $f$-collision, $w \ne 0$, and $D(x) = \bot$.
    Let $R$ be the projection onto the span of states such that $D$ does not contain an $f$-collision, $w \ne 0$, and $D(x) \ne \bot$.
    Let $S$ be the projection onto the span of states such that $D$ does not contain an $f$-collision and $w=0$.
    Observe that $P+Q+R+S=I$.

    Applying $P \circ \CPhsO$ to any state $\sum_{x,w,z,D} \alpha_{x,w,z,D} \ket{x,w,z} \otimes \ket{D}$ in the support of $Q$ yields
    \[
        \sum_{x,w,z,D} \alpha_{x,w,z,D} \ket{x,w,z} \otimes 2^{-n/2} \sum_{r : \exists (x',r') \in D \text{ s.t. } f(x;r)=f(x';r')} (-1)^{w \cdot r} \ket{D \cup (x,r)}.
    \]
    Let $\ket{D} = \ket{(x_1, r_1), \dots, (x_{q'}, r_{q'})}$, where $q' \le q$.
    We can then write the above state as $2^{-n/2} \sum_{i=1}^q \ket{\phi_i}$, where
    \[
        \ket{\phi_i} = \sum_{x,w,z,D} \alpha_{x,w,z,D} \ket{x,w,z} \otimes \sum_{r : f(x; r) = f(x_i; r_i)} (-1)^{w \cdot r} \ket{D \cup (x,r)}.
    \]
    The $\ket{\phi_i}$ are orthogonal and satisfy $\norm{\ket{\phi_i}} \le \sqrt{2^{n-\tau}} \norm{Q\ket{\psi}}$, because there are at most $2^{n-\tau}$ choices of $r$ satisfying $f(x;r) = f(x_i; r_i)$ due to the min-entropy condition on $f$.
    Therefore, $\norm{P \circ \CPhsO \circ Q \ket{\psi}} \le \sqrt{q/2^\tau} \norm{Q \ket{\psi}}$.

    For states $R \ket{\psi} = \sum_{x,w,z,D} \alpha_{x,w,z,D} \ket{x,w,z} \otimes \ket{D}$ in the support of $R$, let $D'$ be $D$ with $x$ removed and write $D = D' \cup (x, r)$ for $r = D(x)$.
    Then $\CPhsO \ket{x,w,z} \otimes \ket{D}$ is
    \[
        \ket{x,w,z} \otimes \left((-1)^{w \cdot r} \left(\ket{D' \cup (x,r)} + \frac{1}{\sqrt{2^n}} \ket{D'}\right) + \frac{1}{2^n} \sum_{r'} (1-(-1)^{w \cdot r} - (-1)^{w \cdot r'}) \ket{D' \cup (x,r')}\right),
    \]
    so
    \[
        P \circ \CPhsO \circ R \ket{\psi} = \sum_{x,w,z,D',r} \alpha_{x,w,z,D',r} \ket{x,w,z} \otimes \frac{1}{2^n} \sum_{r' : \exists (x'',r'') \in D \text{ s.t. } f(x;r')=f(x'';r'')} (1-(-1)^{w \cdot r} - (-1)^{w \cdot r'}) \ket{D' \cup (x,r')}.
    \]
    Let $\ket{D'} = \ket{(x_1, r_1), \dots, (x_{q'}, r_{q'})}$ and
    \[
        \ket{\phi_i} = \frac{1}{2^n} \sum_{x,w,z,D',r} \alpha_{x,w,z,D',r} \ket{x,w,z} \otimes \sum_{r' : f(x; r') = f(x_i; r_i)} (1-(-1)^{w \cdot r} - (-1)^{w \cdot r'}) \ket{D' \cup (x, r')}.
    \]
    Then the $\ket{\phi_i}$ are orthogonal and
    \begin{align*}
        \norm{\ket{\phi_i}}^2 &= \frac{1}{4^n} \sum_{x,w,z,D',r' : f(x; r') = f(x_i; r_i)} \abs{\sum_{r} \alpha_{x,w,z,D',r} (1-(-1)^{w \cdot r} - (-1)^{w \cdot r'})}^2 \\
        &\le \frac{3}{2^n} \sum_{x,w,z,D',r,r' : f(x; r') = f(x_i; r_i)} \abs{\alpha_{x,w,z,D',r}}^2 \\
        &\le \frac{3}{2^\tau} \norm{R \ket{\psi}}^2,
    \end{align*}
    so $\norm{P \circ \CPhsO \circ R \ket{\psi}} \le \sqrt{3q/2^\tau} \norm{R \ket{\psi}}$.

    Finally, $\norm{P \circ \CPhsO \circ P \ket{\psi}} \le \norm{P \ket{\psi}}$ and $\CPhsO \circ S \ket{\psi} = S \ket{\psi}$.
    Putting everything together, $\norm{P \ket{\psi}}$ increases by at most $O(\sqrt{q/2^n})$ with each query.
    Therefore, after $q$ queries, the total norm of $P \ket{\psi}$ is at most $O(\sqrt{q^3/2^n})$, completing the proof.
\end{proof}

The following is similar to Lemma 7 from \cite{Zha19}.

\begin{claim} \label{claim:almost-commuting}
    Consider a quantum system over $x, w, D, y$. Let $f$ be a function such that for every $x$, for uniformly random $r = D(x)$, $f(x; r)$ has min-entropy at least $\tau$.
    Then the following unitaries $O(1/\sqrt{2^\tau})$-almost commute:
    \begin{itemize}
        \item $\CPhsO$, acting on the $x, w, D$ registers.
        \item $\Invert_f$, taking as input the $D, y$ registers and XORing the output onto a new register.
    \end{itemize}
\end{claim}
\begin{proof}
    Recall the definitions of $\CPhsO', \StdDecomp$, and $\Increase$ from \cite{Zha19}.
    In particular, recall that $\CPhsO = \StdDecomp \circ \CPhsO' \circ \StdDecomp \circ \Increase$.
    Since $\Invert_f$ commutes with $\CPhsO'$ and $\Increase$ by construction, it suffices to show that $\Invert_f$ and $\StdDecomp$ are $O(1/\sqrt{2^\tau})$-almost commuting.
    To show this, the following intuition is taken from \cite{Zha19}, adapted to our setting.
    
    For $\StdDecomp$ to have any effect, either (1) $D(x) = \bot$ or (2) $D(x)$ is in uniform superposition; $\StdDecomp$ will simply toggle between the two cases.
    In Case (2), the probability that $\Invert_f$ will find a match at input $x$ is at most $2^{-\tau}$.
    And in Case (1), $\Invert_f$ will never find a match at input $x$.
    Hence, there is an exponentially small error between the action of $\Invert_f$ on these two cases.

    More formally, let $\ket{\psi} = \sum_{x,D,y} \alpha_{x,D,y} \ket{x,D,y}$.
    We omit the $w$ register because neither $\StdDecomp$ nor $\Invert_f$ touch $w$.
    Let $P$ be the projection onto $D(x) = \bot$, and let $Q$ be the projection onto $D(x) = \ket{+}$.
    Writing $D = D' \cup (x,\bot)$,
    \[
        \StdDecomp \circ \Invert_f \circ P \ket{\psi} = 2^{-n/2} \sum_{x,D',y,r} \alpha_{x,D',y} \ket{x,D' \cup (x,r),y,\Invert_f(D',y)}
    \]
    and
    \[
        \Invert_f \circ \StdDecomp \circ P \ket{\psi} = 2^{-n/2} \sum_{x,D',y,r} \alpha_{x,D',y} \ket{x,D' \cup (x,r),y,\Invert_f(D' \cup (x,r),y)},
    \]
    where $\Invert_f(D, y)$ outputs the lexicographically first $x$ such that $f(x,D(x)) = y$, if such an $x$ exists.
    Letting $\Delta := \StdDecomp \circ \Invert_f - \Invert_f \circ \StdDecomp$, we have
    \[
        \Delta \circ P \ket{\psi} = 2^{-n/2} \sum_{x,D',y,r} \alpha_{x,D',y} \ket{x,D' \cup (x,r),y} \otimes \left(\ket{\Invert_f(D',y)} - \ket{\Invert_f(D' \cup (x,r), y)}\right),
    \]
    and
    \begin{align*}
        \norm{\Delta \circ P \ket{\psi}}^2 &= 2^{-n} \sum_{x,D',y,r} \abs{\alpha_{x,D',y}}^2 \left(2 - 2 \cdot 1\{\Invert_f(D',y)=\Invert_f(D' \cup (x,r), y)\}\right) \\
        &= 2 \norm{P \ket{\psi}}^2 - 2 \sum_{x,D',y} \abs{\alpha_{x,D',y}}^2 \Pr_r[\Invert_f(D',y)=\Invert_f(D' \cup (x,r), y)] \\
        &\le \frac{2}{2^{\tau}} \norm{P \ket{\psi}}^2,
    \end{align*}
    since $\Pr_r[\Invert_f(D',y)=\Invert_f(D' \cup (x,r), y)] \ge 1-2^{-\tau}$.
    
    Similarly, $\norm{\Delta \circ Q \ket{\psi}}^2 \le 2 \cdot 2^{-\tau} \norm{Q \ket{\psi}}^2$.
    Since $\Delta \circ (I-P-Q) = 0$, it follows that $\norm{\Delta \ket{\psi}} = O(1/\sqrt{2^\tau})$.
\end{proof}

\bibliographystyle{alpha}
\bibliography{references}

\newcommand{\etalchar}[1]{$^{#1}$}
\begin{thebibliography}{GLR{\etalchar{+}}24}

\bibitem[Aar09]{Aar09}
Scott Aaronson.
\newblock Quantum copy-protection and quantum money.
\newblock In {\em 2009 24th Annual IEEE Conference on Computational
  Complexity}, pages 229--242. IEEE, 2009.

\bibitem[AB24]{AB24}
Prabhanjan Ananth and Amit Behera.
\newblock A modular approach to unclonable cryptography.
\newblock In Leonid Reyzin and Douglas Stebila, editors, {\em Advances in
  Cryptology - {CRYPTO} 2024 - 44th Annual International Cryptology Conference,
  Santa Barbara, CA, USA, August 18-22, 2024, Proceedings, Part {VII}}, volume
  14926 of {\em Lecture Notes in Computer Science}, pages 3--37. Springer,
  2024.

\bibitem[ALL{\etalchar{+}}21]{ALL+21}
Scott Aaronson, Jiahui Liu, Qipeng Liu, Mark Zhandry, and Ruizhe Zhang.
\newblock New approaches for quantum copy-protection.
\newblock In Tal Malkin and Chris Peikert, editors, {\em Advances in Cryptology
  - {CRYPTO} 2021 - 41st Annual International Cryptology Conference, {CRYPTO}
  2021, Virtual Event, August 16-20, 2021, Proceedings, Part {I}}, volume 12825
  of {\em Lecture Notes in Computer Science}, pages 526--555. Springer, 2021.

\bibitem[AP21]{ALP21}
Prabhanjan Ananth and Rolando L.~La Placa.
\newblock Secure software leasing.
\newblock In Anne Canteaut and Fran{\c{c}}ois{-}Xavier Standaert, editors, {\em
  Advances in Cryptology - {EUROCRYPT} 2021 - 40th Annual International
  Conference on the Theory and Applications of Cryptographic Techniques,
  Zagreb, Croatia, October 17-21, 2021, Proceedings, Part {II}}, volume 12697
  of {\em Lecture Notes in Computer Science}, pages 501--530. Springer, 2021.

\bibitem[BB14]{Ben14}
Charles~H Bennett and Gilles Brassard.
\newblock Quantum cryptography: Public key distribution and coin tossing.
\newblock {\em Theoretical computer science}, 560:7--11, 2014.

\bibitem[BGS13]{BGS13}
Anne Broadbent, Gus Gutoski, and Douglas Stebila.
\newblock Quantum one-time programs.
\newblock {\em {IACR} Cryptol. ePrint Arch.}, page 343, 2013.

\bibitem[BI20]{BI20}
Anne Broadbent and Rabib Islam.
\newblock Quantum encryption with certified deletion.
\newblock In Rafael Pass and Krzysztof Pietrzak, editors, {\em Theory of
  Cryptography - 18th International Conference, {TCC} 2020, Durham, NC, USA,
  November 16-19, 2020, Proceedings, Part {III}}, volume 12552 of {\em Lecture
  Notes in Computer Science}, pages 92--122. Springer, 2020.

\bibitem[BS23]{BS23}
Shalev Ben{-}David and Or~Sattath.
\newblock Quantum tokens for digital signatures.
\newblock {\em Quantum}, 7:901, 2023.

\bibitem[CG24]{CG24}
Andrea Coladangelo and Sam Gunn.
\newblock How to use quantum indistinguishability obfuscation.
\newblock In Bojan Mohar, Igor Shinkar, and Ryan O'Donnell, editors, {\em
  Proceedings of the 56th Annual {ACM} Symposium on Theory of Computing, {STOC}
  2024, Vancouver, BC, Canada, June 24-28, 2024}, pages 1003--1008. {ACM},
  2024.

\bibitem[CHV23]{CHV23}
C{\'{e}}line Chevalier, Paul Hermouet, and Quoc{-}Huy Vu.
\newblock Semi-quantum copy-protection and more.
\newblock In Guy~N. Rothblum and Hoeteck Wee, editors, {\em Theory of
  Cryptography - 21st International Conference, {TCC} 2023, Taipei, Taiwan,
  November 29 - December 2, 2023, Proceedings, Part {IV}}, volume 14372 of {\em
  Lecture Notes in Computer Science}, pages 155--182. Springer, 2023.

\bibitem[CLLZ21]{CLLZ21}
Andrea Coladangelo, Jiahui Liu, Qipeng Liu, and Mark Zhandry.
\newblock Hidden cosets and applications to unclonable cryptography.
\newblock In Tal Malkin and Chris Peikert, editors, {\em Advances in Cryptology
  - {CRYPTO} 2021 - 41st Annual International Cryptology Conference, {CRYPTO}
  2021, Virtual Event, August 16-20, 2021, Proceedings, Part {I}}, volume 12825
  of {\em Lecture Notes in Computer Science}, pages 556--584. Springer, 2021.

\bibitem[GKR08]{GKR08}
Shafi Goldwasser, Yael~Tauman Kalai, and Guy~N. Rothblum.
\newblock One-time programs.
\newblock In David~A. Wagner, editor, {\em Advances in Cryptology - {CRYPTO}
  2008, 28th Annual International Cryptology Conference, Santa Barbara, CA,
  USA, August 17-21, 2008. Proceedings}, volume 5157 of {\em Lecture Notes in
  Computer Science}, pages 39--56. Springer, 2008.

\bibitem[GLR{\etalchar{+}}24]{GLR+24}
Aparna Gupte, Jiahui Liu, Justin Raizes, Bhaskar Roberts, and Vinod
  Vaikuntanathan.
\newblock Quantum one-time programs, revisited.
\newblock {\em CoRR}, 2024.

\bibitem[LLQZ22]{LLQZ22}
Jiahui Liu, Qipeng Liu, Luowen Qian, and Mark Zhandry.
\newblock Collusion resistant copy-protection for watermarkable
  functionalities.
\newblock In Eike Kiltz and Vinod Vaikuntanathan, editors, {\em Theory of
  Cryptography - 20th International Conference, {TCC} 2022, Chicago, IL, USA,
  November 7-10, 2022, Proceedings, Part {I}}, volume 13747 of {\em Lecture
  Notes in Computer Science}, pages 294--323. Springer, 2022.

\bibitem[Shm22]{Shm22}
Omri Shmueli.
\newblock Semi-quantum tokenized signatures.
\newblock In Yevgeniy Dodis and Thomas Shrimpton, editors, {\em Advances in
  Cryptology - {CRYPTO} 2022 - 42nd Annual International Cryptology Conference,
  {CRYPTO} 2022, Santa Barbara, CA, USA, August 15-18, 2022, Proceedings, Part
  {I}}, volume 13507 of {\em Lecture Notes in Computer Science}, pages
  296--319. Springer, 2022.

\bibitem[Unr16]{Unr16}
Dominique Unruh.
\newblock Computationally binding quantum commitments.
\newblock In Marc Fischlin and Jean{-}S{\'{e}}bastien Coron, editors, {\em
  Advances in Cryptology - {EUROCRYPT} 2016 - 35th Annual International
  Conference on the Theory and Applications of Cryptographic Techniques,
  Vienna, Austria, May 8-12, 2016, Proceedings, Part {II}}, volume 9666 of {\em
  Lecture Notes in Computer Science}, pages 497--527. Springer, 2016.

\bibitem[Wie83]{Wie83}
Stephen Wiesner.
\newblock Conjugate coding.
\newblock {\em ACM Sigact News}, 15(1):78--88, 1983.

\bibitem[Zha19]{Zha19}
Mark Zhandry.
\newblock How to record quantum queries, and applications to quantum
  indifferentiability.
\newblock In Alexandra Boldyreva and Daniele Micciancio, editors, {\em Advances
  in Cryptology - {CRYPTO} 2019 - 39th Annual International Cryptology
  Conference, Santa Barbara, CA, USA, August 18-22, 2019, Proceedings, Part
  {II}}, volume 11693 of {\em Lecture Notes in Computer Science}, pages
  239--268. Springer, 2019.

\end{thebibliography}

\end{document}